\documentclass[12pt]{article}
\usepackage{cite}
\usepackage{graphicx}
\usepackage{caption}
\usepackage{refstyle}
\usepackage{amsmath}
\usepackage{amsthm}
\usepackage{amsmath}
\usepackage{amsfonts}
\usepackage{amssymb}
\usepackage{marvosym}
\usepackage{bbm}
\usepackage{amsmath,array}
\usepackage{mathtools}
\usepackage{slashbox}
\usepackage[all]{xy}
\usepackage{mathtools}
\usepackage[utf8]{inputenc}
\usepackage[english]{babel}
\usepackage[T3,T1]{fontenc}
\usepackage{fancyhdr}
\DeclareSymbolFont{tipa}{T3}{cmr}{m}{n}
\DeclareMathAccent{\invbreve}{\mathalpha}{tipa}{16}

\usepackage[left=1.6cm,right=1.4cm,top=1.8cm,bottom=1.8cm]{geometry}
\usepackage{setspace}
\newtheorem{theorem}{Theorem}

\title{Achieving Capacity Region of 2-users Weak GIC by Enlarging the Core in a Nested Set of Polymatroids\\
{\large in continuation of arXiv:2012.07820 entitled:
``Optimality of Gaussian in Enlarging HK  Rate Region, and its Overlap with the  Capacity Region of 2-users GIC"} 
\footnote{This draft  provides derivation details for an earlier submission reported in \cite{HK1} (also see~\cite{HK2}).}}

\usepackage{authblk}
\author[1]{Amir K. Khandani\thanks{E\&CE Department, University of Waterloo, Waterloo, Ontario, Canada, khandani@uwaterloo.ca.}}

\begin{document}

\vspace{-4cm}
\maketitle

\vspace{-1cm}
\begin{abstract}
This article shows that  achieving capacity region of a 2-users weak Gaussian Interference Channel (GIC) is equivalent to enlarging the core in a nested set of polymatroids (each equivalent to capacity region of a multiple-access channel) through maximizing a minimum rate, then projecting along its orthogonal span and continuing recursively.  This formulation relies on  defining dummy private messages to capture the effect of interference in GIC. It follows that relying on independent Gaussian random code-books is optimum, and the corresponding solution corresponds to achieving the boundary in HK constraints. 
\end{abstract}

This article provides an alternative view to the results reported in \cite{HK1}.  It is a continuation of an earlier report~\cite{HK2}.  It provides an alternative method to prove optimality of independent Gaussian density functions in generating random code-books in 2-users GIC, initially reported in \cite{HK1}. For the sake of simplicity,  discussions are focused on the weak interference, i.e.,  $a< 1$ and $b<1$, and  
on the case that both users have both private and public messages (see \cite{HK1} for related conditions). Other cases can be handled using the arguments presented here by setting some of the code-books to be empty in the formulation.   

\section{Introduction}

\subsection{Optimality of Gaussian Random Code-books in All-Private Messaging}
Let us assume, in the optimum solution to a 2-users GIC, power allocated to private messages are equal to $\hat{P}_1(\mu)$ and $\hat{P}_2(\mu)$, respectively.
Nested optimality condition derived in \cite{HK1} states that the solution to the same 2-users GIC with power constraints $P_1=\hat{P}_1(\mu)$ and $P_2=\hat{P}_2(\mu)$ will be composed of private messages only. Note that $\hat{P}_1(\mu)$ and $\hat{P}_2(\mu)$ are the lowest power levels for which relying on private messages only results in an optimum solution. 
Another perspective is to allow for private messages only. In this case, the power budget for each private messages will be the entire available power, i.e.,  $P_1$ and $P_2$, respectively, however, as we will see, one of the two transmitters may use less than its corresponding maximum available power. For $\mu\leq 1$, it will be transmitter 2 which may not send at full power $P_2$. For $\mu\geq 1$, it will be transmitter 1 which may not send at full power $P_1$. In any event, the inclusion of constraints on power renders the problem intractable. For this reason, we rely on the parameter of entropy that relates to power and random coding density function. We will be dealing with channels with noise terms that are independent of channel input and are composed of additive Gaussian noise plus interference. In any such channel, a straightforward application of data processing theorem confirms that, the lower is the entropy of  noise plus interference, the higher will be the corresponding capacity value.  

{\bf A review of Entropy Power Inequality:} Consider two random variables; $A$ has an entropy $H_A$ and a power $P_A$, and $G$ is Gaussian $\mathcal{N}(0,1)$. We have

\begin{equation} 
e^{2\breve{H}_A}+ e^{\log_2(2\pi e)} \leq
e^{2H\left(A+G\right)}  \leq  e^{\log_2\left(2\pi e(P_A+1)\right)},
\label{EEn1}
\end{equation}
where $H(\alpha)$ specifies the entropy of the random variable $\alpha$. The left hand side, which is an expression of Entropy Power Inequality,  is satisfied with equality if $A$ is a Gaussian of entropy $H_A$, and consequently of power 
\begin{equation} 
\tilde{P}_A=\frac{1}{2\pi e} 2^{2H_A}.
\label{EEn2}
\end{equation}
This case  means among all densities with with a limit on entropy, Gaussian results in minimizing the entropy of the sum, $A+G$.

The right hand-side, which is an expression of maximum entropy property of Gaussian, is satisfied with  equality if $A$ is Gaussian of power $P_A$. This case means among all densities with a limit on power, Gaussian results in maximizing the entropy of the sum, $A+G$.

Note that, if $\tilde{P}_A=P_A$, or equivalently, 
$H_A=0.5\log_2(2\pi e P_A)$, then the upper bound and the lower bound in \ref{EEn1} coincide. This case  means, in the sense of the two limits in \ref{EEn1},  Gaussian is the minimizer and also the maximizer of the entropy of $A+G$, since minimum and maximum values coincide.  $\blacksquare$ 

Private messages result in a two-inputs channel, maC1, at $Y_1$ and  a two-inputs channel, maC2, at $Y_2$. The two inputs at $Y_1$ are the 
private message of user 1 and the interference caused by the private message of user 2, and likewise for $Y_2$.
Let us consider maC1 and maC2, each composed of two layers: $L^{(1)}_1$, $L^{(1)}_2$ for maC1 and  $L^{(2)}_1$, $L^{(2)}_2$ for maC2 (see Fig.~\ref{Fig-n}).
  \begin{figure}[h]
   \centering
   \includegraphics[width=0.4\textwidth]{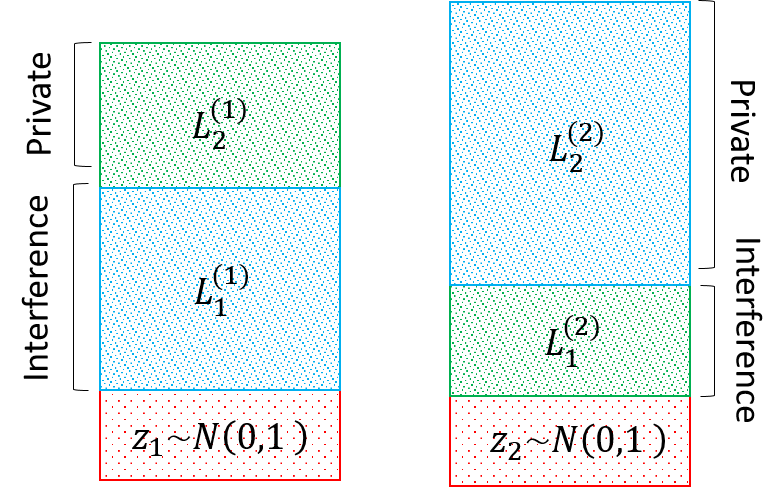}
   \caption{Layers involved in all-private messaging. }
   \label{Fig-n}
 \end{figure}

Let us assume in the optimum solution, we have $H\left(L^{(1)}_1\right) =\breve{H}_1$ and $H\left(L^{(2)}_1\right)  = \breve{H}_2$.  
Let us write the optimization problem by explicitly adding these two constraints.
 \begin{align} \label{La0}
\mbox{Maximize}~~~ & R_{ws}=R(L^{(1)}_2)+\mu R(L^{(2)}_2) \\ 
\mbox{where}~~ & \\ \label{La2}
R\left(L^{(1)}_2\right)&=H\left(L^{(1)}_2+L^{(1)}_1+z_1\right)-
H\left(L^{(1)}_1+z_1\right) \\ \label{La3}
R\left(L^{(2)}_2\right)&=H\left(L^{(2)}_2+L^{(2)}_1+z_2\right)
-H\left(L^{(2)}_1+z_2\right)\\  \label{La4} 
H\left(L^{(1)}_1\right)& =\breve{H}_1 \\ \label{La5}
H\left(L^{(2)}_1\right)&  = \breve{H}_2.
\end{align}
The  optimum solution subject to \ref{La0} to \ref{La5} will be the same as the optimum solution to \ref{La0} to \ref{La3}.
Now let us relax the problem. Relaxation is based on allowing the probability density functions for the four code-layers $L^{(1)}_1$, $L^{(1)}_2$,  $L^{(2)}_1$ and $L^{(2)}_2$ to be independent of each other. Selecting the random coding density for maximization of \ref{La0} narrows down to maximization of the entropy values $H\left(L^{(1)}_1+L^{(1)}_2+z_1\right)$ and  $H\left(L^{(2)}_1+L^{(2)}_2+z_2\right)$, and minimizing the  entropy values of 
$H\left(L^{(1)}_1+z_1\right)$ and $H\left(L^{(2)}_1+z_2\right)$ subject to \ref{La4} and \ref{La5} and power constraints. Assuming $z_1,z_2$ are $\mathcal{N}(0,1)$ and applying entropy power inequality to $L^{(1)}_1+z_1$ and to $L^{(2)}_1+z_2$, it follows 
 
\begin{align} \label{En1}
 e^{2H\left(L^{(1)}_1+z_1\right)}& ~\geq~e^{2\breve{H}_1}+ e^{\log_2(2\pi e)} \\ \label{En2}
 e^{2H\left(L^{(2)}_1+z_2\right)}& ~\geq~  e^{2\breve{H}_2}+ e^{\log_2(2\pi e)}.
\end{align}
with equality if $L^{(1)}_1$ and $L^{(2)}_1$ are Gaussian. Note that equality conditions in \ref{En1} and \ref{En2} result in minimizing the entropy of $H\left(L^{(1)}_1+z_1\right)$ and $H\left(L^{(2)}_1+z_2\right)$, respectively. Relying on the relationship between entropy and power in a Gaussian random variable, the power of interference terms and consequently the power of code-book used by each transmitter will be determined. Given the power values, the entropy of signal terms, i.e., $H\left(L^{(1)}_1+L^{(1)}_2+z_1\right)$ in \ref{La2} and $H\left(L^{(2)}_1+L^{(2)}_2+z_2\right)$ in \ref{La3} are maximized, since all random variables in summations constructing the two signals are Gaussian, each with its maximum allowable power value. Note that one of the two transmitters may use only part of its power budget, and in this case, the allowable power value for that transmitter is governed by some other considerations to be discussed.   

 A question remains  concerning the possibility of realizing the target entropy levels $\breve{H}_1$ and $\breve{H}_2$ using Gaussian density functions, and its impact on the signal terms. The answer to the question is positive since Gaussian requires the least amount of power for a given entropy. In other words, if some other density function has been capable of generating entropy values $\breve{H}_1$ and $\breve{H}_2$ within the available power budgets, then Gaussian density can do it as well. By selecting Gaussian density functions to realize the target entropy levels, we are maximizing the individual capacity values in \ref{La2} and \ref{La3} by minimizing the entropy of their respective interference plus noise terms. For example, in \ref{La2}, $L^{(1)}_2$ is the signal and  $L^{(1)}_1+z_1$ (a term independent from the corresponding signal) is the noise-plus-interference. According to data processing theorem, a channel with a noise of a lower entropy will have a higher capacity, regardless of the impact of noise in signal term, i.e., in $L^{(1)}_2+L^{(1)}_1+z_1$. On the other hand, the power of interference term $L^{(1)}_1$ affects the second capacity term in \ref{La3}. Indeed, in this setup, as we are avoiding the use of public messages, one of the two transmitters may not use its entire power budget. The important point is that the transmitter sending less than full-power should rely on a Gaussian density because it minimizes the impact of the interference term on the capacity achieved by the transmitter sending at full power, and at the same time, it maximize the entropy of both signal terms subject to the allowable power values used by each transmitter.  
Expressing \ref{La0} as, 

 \begin{align} \label{C1}
R_{ws}& = 
 H\left(L^{(1)}_2+L^{(1)}_1+z_1\right)-H\left(L^{(1)}_1+z_1\right)
+\mu \left[H\left(L^{(2)}_2+L^{(2)}_1+z_2\right)
-H\left(L^{(2)}_1+z_2\right)\right]=A+B \\ \nonumber 
&\mbox{where} \\ \label{C2}
& A =H\left(L^{(1)}_2+L^{(1)}_1+z_1\right)-\mu H\left(L^{(2)}_1+z_2\right),
 ~~B=\mu H\left(L^{(2)}_2+L^{(2)}_1+z_2\right)-H\left(L^{(1)}_1+z_1\right).
\end{align}
Note that the power of interference term $L^{(2)}_1$, denoted as $P\left(L^{(2)}_1\right)$, satisfies (see Fig.~\ref{fig1} for definition of cross channel gains $a$ and $b$), 
\begin{equation}
P\left(L^{(2)}_1\right)=bP\left(L^{(1)}_2\right),
\label{Rn1}
\end{equation}
and $P\left(L^{(1)}_1\right)$ satisfies
\begin{equation}
P\left(L^{(1)}_1\right)=aP\left(L^{(2)}_2\right).
\label{Rn2}
\end{equation} 
Replacing for entropy values of Gaussian terms, and then computing derivatives of \ref{C1} and \ref{C2} with respect to the power values (subject to \ref{Rn1} and \ref{Rn2}), it is straightforward to show that $A$ and/or $B$ in \ref{C2} will be maximized by having its corresponding transmitter to transmit at full power, i.e., at least one of transmitters uses all its power budget.  

Reference \cite{HK1} shows that the lower part of the boundary curve, i.e., $\mu\leq 1$, for 2-users weak GIC starts from a corner point in which user 1 sends an all-private message and user 2 sends an all-public message. The case considered in this section does not permit the use of public messages, so the corner point in an all-private signaling will correspond to user 1 sending a private message at full-power and user 2 remaining silent. Then, in \cite{HK1},  moving counterclockwise along the lower part of the boundary, the power allocated to private message of user 2 gradually increases, while the power allocated to its public message decreases. However, up to point $D_3$  (see Fig. 11 of \cite{HK1}) user 1 continues sending an all-private message at full power. 
In an all-private signaling scheme presented in this section, the power allocated to private messages will be the same as the case of \cite{HK1}, while public messages accompanying  the private messages in \cite{HK1} will be dropped. This allows moving from corner point $A$ to point $D_3$ by keeping the power of private messages the same as in \cite{HK1}, while putting the power of accompanying public messages (involved in forming Fig. 11 of \cite{HK1}) equal to zero.     An all-private boundary curve does not extend beyond the counterpart  of point $D_3$ (counterpart means the same point without public messages). 

The arguments presented in this section, in terms of using Gaussian density for interference term in order to minimize the entropy of the noise-plus-interference, can be better understood by expressing the problem of computing the lower part of the boundary curve as follows:

\begin{align} \label{CCC1}
\mbox{Maximize}~~& R\left(L^{(1)}_2\right)=H\left(L^{(1)}_2+L^{(1)}_1+z_1\right)-
H\left(L^{(1)}_1+z_1\right) \\ \label{CCC2}
\mbox{Subject to:}~~& R\left(L^{(2)}_2\right)=H\left(L^{(2)}_2+L^{(2)}_1+z_2\right)
-H\left(L^{(2)}_1+z_2\right)\equiv \mathsf{C}_2.
\end{align}
In covering the lower part of the boundary, the value of $\mathsf{C}_2$ will gradually increase in infinitesimal steps. As the objective is to maximize $H\left(L^{(1)}_2+L^{(1)}_1+z_1\right)-
H\left(L^{(1)}_1+z_1\right)$, the step for increasing $\mathsf{C}_2$ to $\mathsf{C}_2+\delta \mathsf{C}_2$ should be such that the entropy of the noise term $H\left(L^{(1)}_1+z_1\right)$ is minimized, which justifies the arguments starting from \ref{En1}  and concluding $L^{(1)}_1$ should be Gaussian. Using Gaussian code-books for other layers follow knowing $L^{(1)}_1$ is Gaussian. A similar argument applies to the upper part of the boundary, i.e., $\mu\geq 1$, using \ref{En2} and concluding interference term $L^{(2)}_1$ should be Gaussian. 

In \ref{La4} and \ref{La5}, the condition on the density functions are expressed in terms of corresponding  entropy values as parameters. Let us relax the problem (to obtain an upper bound) where the probability densities of four layers: $L^{(1)}_1$, $L^{(1)}_2$, $L^{(2)}_1$, $L^{(2)}_2$ are allowed to be selected independently. Using Gaussian density for $L^{(1)}_1$ and $L^{(2)}_1$ minimizes the entropy of additive noise-plus-interference in channels formed at $Y_1$ and $Y_2$ (see \ref{En1} and \ref{En2}).  Note that, in selecting a density function to realize certain entropy value, one needs to adhere to existing constraints on power. Since Gaussian requires least amount of power for a given entropy, if there exist densities satisfying 
\begin{align}
H\left(L^{(1)}_1\right) = \breve{H}_1 \\ 
H\left(L^{(2)}_1\right)  = \breve{H}_2.
\end{align}
subject to existing power constraints, then Gaussian density can also achieve the target entropy values within the allowable power limits. With Gaussian interference-plus-noise, the density of $L^{(1)}_2$ and $L^{(2)}_2$, which are each the input to a channel with additive Gaussian noise-plus-interference (noise-plus-interference term in each case is independent from corresponding channel input) should be Gaussian. A remaining question concerns the role of input power constraint in selecting the density functions for channel inputs, i.e., $L^{(1)}_2$ and $L^{(2)}_2$. In the case of  $L^{(1)}_2$, it is always at full power, meaning that the lower and upper bound in \ref{EEn1} coincide. 

A remaining point is to justify replacing the role of density function for $L^{(2)}_2$ in \ref{CCC2} with the entropy value. The problem is the relationship between density function and entropy value is not one-to-one.  Moving along the boundary can be also formulated in terms of an infinitesimal increase in the entropy of $H\left(L^{(1)}_1\right) =\breve{H}_1+\delta H$ in \ref{La4}.  
In the context of all-private messaging discussed here, a $\delta H$ step in entropy means increasing   the power of private message of user 2 by $\delta P$,  while user 1 continues to transmit at full power. For a Gaussian $L^{(1)}_1$, a fixed $\delta H$ is accompanied with minimum $\delta P$, i.e., minimum increase in the power of additive noise-plus-interference in \ref{CCC1}. In other words, use of Gaussian density in view of \ref{En1} (for minimizing the entropy of the additive noise-plus-interference) has also kept the power of additive noise-plus-interference at $Y_1$ minimum.   
In other words, in \ref{CCC2}, optimality condition necessitates realizing  $\mathsf{C}_2$ with minimum power for $L^{(2)}_2$. If one could achieve the value of $\mathsf{C}_2$ with a lower power for  $L^{(2)}_2$, it would translates to a lower power for the additive noise-plus-interference term at $Y_1$, which could be used to increase $R\left(L^{(1)}_2\right)$, contradicting the optimality of the solution. Gaussian $L^{(2)}_1$ minimizes the entropy of the interference-plus-noise at $Y_2$ (see \ref{En2}) and then using Gaussian $L^{(2)}_2$ would enable realizing the target capacity value of $\mathsf{C}_2$ with minimum power. Note that all these discussions are focused on $\mu\leq 1$, where $X_1$ sends at full power, i.e., $P_1$, and $X_2$ may send at a power level less than $P_2$. 

Noting above discussions, it is concluded that in the optimum solution to \ref{CCC1}, \ref{CCC2} (while relaxing consistency condition, i.e., allowing independent density functions for $L^{(1)}_1$, $L^{(1)}_2$, $L^{(2)}_1$, $L^{(2)}_2$), the consistency conditions have been automatically satisfied. This means, the optimum solution to the upper bound coincides with the optimum solution to the original problem.

\subsection{Considering both Public and Private Messages: Preliminaries}

Consider a memory-less multiple-access channel with inputs $X_1,\ldots,X_M$ and output $Y$ described by conditional distribution $f_{Y|X_1,\ldots,X_M}(y|x_1,\ldots,x_M)$, where fixed and independent density functions are deployed for generating random code-books. The capacity region of such a channel forms a polymatroid~\cite{HK3}. Sum-capacity facet is composed of $M!$ corner points which can be realized relying on different orderings in successive decoding of messages.  
The optimum solution in maximizing a weighted sum-rate falls on a corner point on the sum-capacity facet, which, through ordering of messages in successive decoding, assigns higher rate values to terms with a larger weight in the  weighted sum-rate (see Lemma 3.2 of ~\cite{HK3}). Reference \cite{HK4} shows that the fairest corner point, maximizing the minimum rate in a memory-less multiple-access channel, corresponds to ordering users according to their rate values in increasing order from bottom to top and applying successive decoding from top to bottom. 

Consider the problem of transmitting 
$M$ messages over a multiple-access channel with additive 
Gaussian noise $\mathcal{N}(0,1)$.
The solution is the set of all rate vectors $\mathsf{R}$ satisfying:
\begin{equation}
\mathsf{R}(S)\leq I\left(Y;X_i, i\in S~|~X_i, i\notin S\right)~~\forall S\subset\{1,2,\ldots,M\}~~\mbox{where}~~\mathsf{R}(S)=\sum_{i\in S} r_i.
\label{M1}
\end{equation}
The goal is to maximize a weighted sum-rate 
$L=\sum_{i=1}^M \lambda_i r_i$.  The power used for sending message 
$i$ is denoted as $p_i$. Instead of imposing a fixed limit on $p_i, i=1,\ldots,M$, the vector of power budgets, $\mathbf{p}=(p_1,\ldots,p_M)$, is limited by a set of linear constraints with positive coefficients. This results in a simplex region in the positive orthant, called a Power Vector Simplex, or $\mathsf{PVS}_M$, hereafter. 

Expressions in \ref{M1} are identical to the case of superposition coding with successive decoding. The difference between the two formulations stems from differences in constraints on power values. There are $2^M-1$ expressions defining \ref{M1}, with left hand sides composed of different sum-rates. Each rate value appears in $2^{M/2}$ expressions. We are interested in maximizing a weighted sum-rate, $R_{ws}$, over \ref{M1}. The optimum solution will be a vector of rates that maximizes $R_{ws}$ subject to the underlying power constraints. Constraints on power are a set of linear (equality) terms over power values allocated to different code-books. Assuming the optimum value is equal to $R^{opt}_{ws}$, the optimization problem can be expressed as realizing $R_{ws}=R^{opt}_{ws}$ using a power vector that satisfies the power constraints and is Pareto minimal, i.e., a decrease in any of its power values results in reducing the achieved $R_{ws}$. In other terms, the optimum rate vector maximizes $R_{ws}$ with minimum power.  
Given an optimum rate vector, the order of code-books can be changed if the set of constraints on their respective power values permit. In such a case, there will be multiple optimum solutions corresponding to permuting the positions of code-books in successive decoding, and accordingly adjusting their power levels to achieve the same optimum rate vector with  shuffling of its components. For example, if the constraint is on total power in a superposition code composed of $M$ code-books, if a vector of rate $\vec{r}$ is achievable, different rate vectors obtained by permuting the elements of $\vec{r}$ will be achievable as well. In our case, constraints on power are such that only one order for successive decoding can achieve $R^{opt}_{ws}$. In this case, the optimization problem is equivalent to minimizing the power values for realizing a desired rate vector subject to the optimum decoding order.

Let us use notation $\mathsf{PVS}_{M-K}(p_1,\ldots,p_K)$ to refer to power simplex available to massages indexed by $K+1$ to $M$ if the power vector for messages indexed by $1,\ldots,K$ is $(p_1,\ldots,p_K)$. Let us focus on the rate of the first message (decoded last), and denote it by $r_1$. For a given value of $r_1$, placing the message first minimizes the required $p_1$, and thereby enlarges the Power Vector Simplex available to remaining messages, i.e.,  
\begin{equation}
\mathsf{PVS}_{M-1}(\breve{p}_1)\subset \mathsf{PVS}_{M-1}(\overline{p}_1)~~~\mbox{if}~~~\overline{p}_1<\breve{p}_1.
\label{Eqq2}
\end{equation}
Minimizing $p_1$ also reduces the amount of interference observed by messages   $2,\ldots,M$, which are all ordered above the first message. More generally, let us assume, as part of problem definition, the  rates of the first $\ell$ messages (decoded last) are set at $r_1$ to $r_{\ell}$ .  Regardless of the form of power simplex, the optimum random coding for all layers will be independent Gaussian. Using a Gaussian density for a layer  results in:  (i) worst noise for layers above it, (ii) least amount of power allocated to the layer, thereby maximally enlarging the $\mathsf{PVS}$ governing the power vectors available for remaining code-books above it, (iii) the power of interference is minimized. Regardless of values of $r_1$ to $r_{\ell}$ and the structure of $\mathsf{PVS}$, the positive impacts in (ii) and (iii) are more advantageous than the negative impact in (i). Also see Theorems \ref{Gau}, \ref{Gau2} and Example 1.

{\bf Remark:} Let us consider maximizing a weighted sum-rate over a multiple-access channel (or equivalently in a superposition code). The power constraints are such that one of the non-zero rates in the optimum rate vector has zero weight in the weighted sum-rate.  The corresponding rate value contributes to the formation of optimal solution by appearing in the left hand sides of  $2^{M/2}$ expressions among the $2^M-1$ expressions in \ref{M1}. This means, it indirectly affects the achieved optimum value for the weighted sum-rate.  Solution of 2-users GIC narrows down to minimizing the power required to realize all non-zero rate values, which, directly or indirectly, affect the value of the weighted sum-rate. 
$\blacksquare$ 

Following Theorem summarizes the above arguments. 
\begin{theorem}
\label{NNN}
Consider the problem of maximizing a weighted sum-rate over a multiple-access region, or over the intersection of several multiple-access regions, subject to power constraints captured by certain $\mathsf{PVS}$. Assume there is a set, $S_R$, of  rate values with zero weight in the weighted sum-rate, which are assigned a non-zero rate value in the optimizing solution. This can be the case only if, subject to the set of equations corresponding to active constraints among the set in \ref{M1}, the rate values with non-zero weight turn out to be an increasing function of rate values in $S_R$. 
\end{theorem}

\begin{proof}
Having a non-zero rate value for a code-book always comes at a cost in terms of power. Accordingly, power constraints permit allocation of power to code-books that directly or indirectly result in improving the achieved weighted sum-rate. For terms that have zero weight in the weighted sum-rate, the relevance of their rate value to improvement in the weighted sum-rate is indirectly embedded within the relationships imposed by the set of active constraints among expressions in \ref{M1}.  On the other hand, the right hand sides of the set of active constraints are a set of mutual information terms. Optimization of random code-books and power allocation result in minimizing the power (subject to the underlying  $\mathsf{PVS}$) required to achieve the rates  that directly or indirectly affect the weighted sum-rate. 
\end{proof}
Following Theorem will be beneficial is establishing the result of the Section \ref{sec1.2}.

\begin{theorem}
\label{NNN2}
Consider the problem of constructing a superposition code over an AWGN channel, $\mathcal{N}(0,1)$, with $2$ layers. The objective is to maximize the rate of the second layer. The total power is fixed at $\breve{P}$.  The problem is linked to a second problem which: (i) does not allow allocating zero power to the first layer, and (ii) the power allocated to the first layer will be minimized if the corresponding code-book is Gaussian. Use of Gaussian code-books for both layers is optimum. 
\end{theorem}
\begin{proof}
The condition that  ``power allocated to the first layer will be minimized if corresponding  code-books is Gaussian'' is equivalent to the condition that the ``mutual information between first layer and output should be larger than or equal to a given value'', say $\breve{r}_1$. Rate of the second layer is bounded by 
\begin{equation}
\label{G11}
r_2\leq 0.5\log_2(1+\breve{P})-\breve{r}_1.
\end{equation} 
The condition for equality in \ref{G11} is that both layers are Gaussian. 
\end{proof}
Theorem \ref{NNN2} can be generalized to the case of maximizing a weighted sum-rate over a superposition code (or a multiple-access channel) with $M$ messages.  Weighted sum-rate is defined over layers $\ell+1,\ldots,M$, and the power allocated to layers $1,\ldots,\ell$ will be minimized if their code-books are Gaussian. It follows that using independent Gaussian code-books for all layers is optimum. 
As will be shown in Section \ref{sec1.2}, in the case of 2-users GIC, there are two such superposition codes that are linked together. Each code-book in one superposition code has a pair code-book in the second one. Random coding alphabets in code-books forming a pair are obtained by scaling the output of the same generator.  Weighted sum-rate involves layers $\ell+1,\ldots,M$, $\ell=M/2$, in the two superposition codes.  The code-books in layers $1,\ldots,\ell$ in one superposition code are paired with code-books in layers $\ell+1,\ldots,M$ in the other one, and vice versa. Noting Theorem \ref{NNN2}, use of Gaussian code-books is optimum in terms of maximizing the weighted sum-rate over the two superposition codes. 

\subsection{Construction of Code-books Relying on Infinitesimal Layers} \label{sec1.2}
Complexity of solving \ref{L1} to \ref{L13} is partially due to the inclusion of power allocation problem. To overcome this issue, in \cite{HK1}, the power budget $P_1$ and $P_2$ are each divided into a continuum of $L_1=P_1/\delta$ and $L_2=P_2/\delta$ infinitesimal values of power $\delta$, respectively, each called an infinitesimal layer. 
$Y_1$ and $Y_2$, each receive all the $L_1+L_2$ infinitesimal layers. There is a  close connection between ``superposition coding in a point-to-point link" and ``relying on separate transmitters in a multiple-access channel". In particular, in both configurations, using multi-level code-books with independent layers in transmitter(s) and successive decoding in receiver(s) is optimum (assuming memory-less channels). 
Relying on this viewpoint, the optimization problem in \ref{L1} to \ref{L13} can be considered as a setting in which $L_1+L_2$ separate transmitters, each with power $\delta$, send their signals to two receivers, $Y_1$ and $Y_2$. 
In this case, $Y_1$ and $Y_2$ will each receive $L_1+L_2$ infinitesimal code-books, with a one-to-one correspondence between code-books at $Y_1$ and those at $Y_2$ (in  the sense that two copies, with different scale factors, of each infinitesimal code-book are received, one at $Y_1$ and the other one at $Y_2$). Optimizing the weighted sum-rate over this setup narrows down to determining the rate of each code layer, and deciding the order of successive decoding at each of the two receivers. Note that here we are dealing with two sets of constraints, each of the form in \ref{M1}. Solving the two sets of constraints of the form in \ref{M1} results in a subset of the corresponding constraints to be satisfied with equality (active constraints).  Active constraints can be solved as a set of equations to determine the rates of each infinitesimal code-book.  
Successive decoding of infinitesimal code-books at $Y_1$ and $Y_2$ proceeds from top to bottom. infinitesimal code-books constructing public messages are decoded first. Then comes the private code-book of user 1 at $Y_1$ and the private code-book of user 2 at $Y_2$. The key point is that, the optimization of power allocation, which is determining the number of layers in public and private messages, will be such that after decoding of infinitesimal public code-books, infinitesimal private code-books of user 1 are decoded at $Y_1$ and infinitesimal private code-books of user 2 are decoded at $Y_2$. For each infinitesimal private code-book at $Y_1$, there is a corresponding copy at $Y_2$ (called dummy private message of user 1) which is placed after all infinitesimal private code-books of user 2 . Likewise, for each infinitesimal private code-book at $Y_2$, there is a corresponding copy at $Y_1$ (called dummy private message of user 2) which is placed after all infinitesimal private code-books of user 1. This ordering is a consequence of optimization problem, and can be explained relying on Theorem \ref{Th1}.    

Consider a pair of infinitesimal code-books $\Omega_1$ and  $\Omega_2$, formed at $Y_1$ and $Y_2$, respectively. Let us assume rate of $\Omega_1$, denoted as $\delta r_1$ contributes to the weighted sum-rate. Using Gaussian random code-books  minimizes the amount of power in $\Omega_1$ for the given rate $\delta r_1$. On the other hand, $\Omega_2$ is formed by the same probability density function after being subject to a scale factor corresponding to relevant channel cross gain. As a result, minimizing the power of $\Omega_1$ will also minimize the power of $\Omega_2$. Assuming $\Omega_1$ is part of the private message of user 1, then $\Omega_2$ will act as interference at $Y_2$. Note that the total power received at $Y_1$ is fixed at $P_1+aP_2$ and total power received at $Y_2$ is fixed at $bP_1+P_2$, as required in Theorem \ref{NNN2}. 
By minimizing power allocated to private message of user 1, the power of interference observed by user 2 is minimized, but at the cost of having a Gaussian density for the interference term. Noting Theorem \ref{NNN2}, the overall trade-off will be in favor of relying on Gaussian density for all infinitesimal code-books forming the private messages of users 1 and 2.  

Following theorem summarizes the above arguments. 

\begin{theorem}
\label{GG}
Construction of code-books relying on infinitesimal Gaussian layers optimizes the weighted sum-rate in 2-users GIC.  
\end{theorem}   

\begin{proof}
Let us focus on user 1. 
Optimality of Gaussian random codes for  infinitesimal code-books is due to the fact that, at $Y_1$, the set of public layers and private layers of user 1 are equivalent to a superposition code with successive decoding subject to AWGN plus interference from dummy private messages of user 2. For each infinitesimal code-book in the private message of user 1, say code-book $\Omega$, there is a counterpart copy, say $\breve{\Omega}$, with the same probability density at $Y_2$ (with some scaling due to relevant channel cross gain). The rate that could be embedded in  $\breve{\Omega}$ is less than the rate of $\Omega$, simply because rate of $\breve{\Omega}$ does not count in the weighted sum-rate. The optimization governing rates of infinitesimal  code-books could set the rate of the corresponding layer at the smaller value corresponding to $\breve{\Omega}$, which would result in the code-book to be decoded at both $Y_1$ and $Y_2$, or set the rate at the higher value corresponding to 
${\Omega}$, which would entail the code-book can be decoded only at $Y_1$. The option of decoding the infinitesimal code-book at both $Y_1$ and $Y_2$ would entail it is indeed part of a public message. This cannot be the case since the optimization problem has determined a fixed (optimum) number of layers to be part of public messages. This is indeed a statement of nested optimality condition. Regardless of how the rate of each infinitesimal code-book is decided, the optimum random code for the construction of all infinitesimal code-books at $Y_1$ as well as at $Y_2$ (when considered independently) is based on using independent Gaussian density of power $\delta$ for all code-books ($\delta$ is the power at the transmit side, it will be subject to channel gain before reaching the receiver side). Since the use of Gaussian code-books for all layers satisfy the consistency requirement (in terms of having the same density for counterpart copies of the same infinitesimal code-books), then Gaussian density will be also optimum when the constraints at $Y_1$ and $Y_2$ are considered together. Note that consistency requirement on random code-books is equivalent to what is expressed in \ref{L10} to \ref{L13} for the case of four code-books. Indeed, the four code-books in \ref{L1} to \ref{L13} are obtained by combining consecutive infinitesimal Gaussian code-books that could form a continuous layer (see \cite{HK1} for definition).  

To shed some lights, let us present the proof in a different language.  Let us consider an upper bound on the solution in which the random coding density function for the two code-books forming a pair can be selected independent of each other. Consider an ordering of layers at $Y_1$ and an independent ordering of layers at $Y_2$.  Consider a strategy in which the random coding density for each layer is selected (independent of its pair) to maximize its rate given its position in successive decoding. Let us consider a pair of layers achieving rates $\breve{r}_1$ and $\breve{r}_2$ at $Y_1$ and $Y_2$, respectively.  Once the rates of layers are independently optimized (using independent Gaussian code-books of power $\delta$ for all layers) a second algorithm decides for  a single rate for the two code-books forming a pair, i.e., for the pair of rates $\breve{r}_1$ and $\breve{r}_2$, the final rate will be selected as $[\breve{r}_1~\mbox{or}~\breve{r}_2]$ depending on the contribution of the code-book to weighted sum-rate. As a result, for some of the pairs (forming private messages) only one of the two code-books will be decodable. However, for all pairs, at least one of the two code-books, called non-redundant component of the pair (vs. redundant component), will be decoded and its rate directly appears in the weight sum-rate. Using Gaussian random code-book is necessary to maximize the rate of the non-redundant component, and as a side-effect, it also maximizes the rate of redundant component. The consistency condition is automatically satisfied since the optimum density function for all layers, regardless of their position in successive decoding, is a Gaussian of power $\delta$ at the transmit side (subject to scaling by channel gains). This also means the upper bound coincides with the actual solution. 

In summary, since each infinitesimal code-book is decoded in at least one of the two receivers, the random coding density for the pair should be the optimum one, which is a Gaussian of power $\delta$. This is possible since a Gaussian of power $\delta$ for all layers  automatically satisfies the consistency requirement, i.e., having the same random coding density for the two code-books forming a pair. 
\end{proof}

\subsection{Formation of Multiple-access Channels with Four Inputs}
Consider the 2-users GIC model in \ref{fig1} for the case   
 that both users have both private and public messages. Let us assume public and private messages of each user are independent (to be justified later). 
With independent $U_1$, $V_1$, $U_2$ and $V_2$, a multiple-access channel, $\overline{M\!AC_1}$ with four inputs 
is formed at $Y_1$, and a multiple-access channel, $\overline{M\!AC_2}$ with four  inputs is formed at $Y_2$.  
Let us define the corresponding optimization problem  $\mathcal{P}_0$ as follows
\begin{align} 
\label{L1}
\mbox{Maximize:}~~ & L_0= R_{U_1}+ R_{V_1} + \mu(R_{U_2}+ R_{V_2})  \\ 
\mbox{Subject to:}~~~ &   \nonumber \\ \label{L2}
\left(R_{U_1}, R_{V_1}, R_{U_2}, R^{(1)}_{V_2}\right)& ~~\in~~  \overline{M\!AC_1}\left(P^{(1)}_{U_1}, P^{(1)}_{V_1},P^{(1)}_{U_2},P^{(1)}_{V_2}\right)\\ \label{L3}
\left(R_{U_1}, R^{(2)}_{V_1}, R_{U_2}, R_{V_2}\right)& ~~\in~~  \overline{M\!AC_2}\left(P^{(2)}_{U_1}, P^{(2)}_{V_1},P^{(2)}_{U_2},P^{(2)}_{V_2}\right)\\
\mbox{where:}~~~ &   \nonumber \\ \label{L4}
P^{(1)}_{U_1}=P^{(2)}_{U_1}/b  & ~~\triangleq~~ P_{U_1}  \\  \label{L5}    
P^{(2)}_{U_2}=P^{(1)}_{U_2}/a & ~~\triangleq~~  P_{U_2}  \\ \label{L6}   
 P^{(1)}_{V_1}=P^{(2)}_{V_1}/b & ~~\triangleq~~   P_{V_1} \\  \label{L7}
P^{(2)}_{V_2}=P^{(1)}_{V_2}/a & ~~\triangleq~~  P_{V_2}   \\  \label{L8} 
P_{U_1} +P_{V_1} & ~~=~~    P_1  \\   \label{L9} 
P_{U_2} +P_{V_2} & ~~=~~   P_2 \\  \label{L10} 
 f^{(1)}_{U_1}=f^{(2)}_{U_1} & ~~\triangleq~~ f_{U_1}    \\  \label{L11}  
 f^{(1)}_{U_2}=f^{(2)}_{U_2}& ~~\triangleq~~ f_{U_2}   \\ \label{L12}   
 f^{(2)}_{V_1}\circeq f^{(1)}_{V_1} & ~~\triangleq~~ f_{V_1}   \\  \label{L13}
f^{(1)}_{V_2} \circeq   f^{(2)}_{V_2} & ~~\triangleq~~ f_{V_2}
\end{align}  
where $\triangleq$ means "denoted as", and $\circeq$ means the density functions will be the same once random variables are normalized by dividing with their respective standard deviations.  In \ref{L2} and \ref{L3}, MAC regions are expressed as a function of the power budget for their respective message set to emphasize the role of power allocation. Solution to \ref{L1} to \ref{L13} should satisfy some consistency requirement, in the sense that  power vectors related to $\overline{M\!AC_1}$ 
and $\overline{M\!AC_2}$  are linked to each other through  \ref{L4} to \ref{L7}, while satisfying \ref{L8} and \ref{L9}, and random coding density functions should be consistent in view of \ref{L10} to \ref{L13}.  

Note that 
\begin{align} \label{Dum1}
R^{(1)}_{V_2}& =I(\sqrt{a}V_2;Y_1|U_1,U_2,V_1) \\ \label{Dum2}
R^{(2)}_{V_1}& =I(\sqrt{b}V_1;Y_2|U_1,U_2,V_2). 
\end{align}
are simply mathematical expressions that could be used to form the expressions governing successive decoding. 

{\bf Dummy Private Messages:} Let us use refer to messages resulting in 
rates $R^{(1)}_{V_2}$ and $R^{(2)}_{V_1}$ in solving $\mathcal{P}_0$ as dummy\footnote{Called dummy since these are merely an expression in terms of power.} private messages. Dummy private messages capture the role of interference in 2-users GIC, and relates its solution to the solution of problem 
$\mathcal{P}_0$. Messages realizing  $R_{V_1}$ and $R_{V_2}$ are called non-redundant private messages.  

To further clarify the roles of dummy private messages and thereby the relationships between 2-users GIC and $\overline{M\!AC_1}$, $\overline{M\!AC_2}$, let us consider the optimum solution to 2-users GIC, and focus on interference terms at $Y_1$ and $Y_2$. These are random code-books that have been already labeled/selected indirectly through private messages. However, consider a hypothetical case in which one could provide an independent set of labels for elements of these code-books up a rate that would be decodable at their respective receiver. This would result in 4 rate values at each receiver. These 4-tuple rate values cannot fall outside  $\overline{M\!AC_1}$ (at $Y_1$) or outside $\overline{M\!AC_2}$ (at $Y_2$). 

Note that, in dealing with $\overline{M\!AC_1}$ and $\overline{M\!AC_2}$, rates of dummy private messages are simply an expression in terms of relevant power values. In 2-users GIC,  the actual code-book/rate of the dummy private message of user 1 is the same as the code-book/rate for the non-redundant private message of user 2 and vice versa. As we will see, formulation (handling non-redundant rates) for 2-users GIC is obtained by projecting $\overline{M\!AC_1}$ and $\overline{M\!AC_2}$ along the orthogonal span of the rate of their relevant dummy private message. This removes the role of the rate of dummy private messages  in the formulation relevant to 2-users GIC, i.e., solely accounts for their impacts in terms of an interference term from a lower level code-book in successive decoding of non-redundant messages. Consequently, dummy private messages (interference terms) do not need to be decodable in 2-users GIC.  $\blacksquare$ 

{\bf Dominant Multiple-Access Region:}
The optimum solution maximizing $L_0$ in \ref{L1}, 
 either falls on the sum-rate front of 
$\overline{M\!AC_1}$, which is equal to
\begin{equation}
0.5\log_2\left(1+(P1+aP_2)/\sigma^2\right),
\label{EQN1}
\end{equation}
and on the boundary of $\overline{M\!AC_2}$, or on the sum-rate front of 
$\overline{M\!AC_2}$, which is equal to
\begin{equation}
0.5\log_2\left(1+(bP1+P_2)/\sigma^2\right),
\label{EQN2}
\end{equation}
and on the boundary of $\overline{M\!AC_1}$.  The MAC with the optimizing 
sum-rate front is called the dominant MAC. Note that sum-rates in \ref{EQN1} and \ref{EQN2} include rates of dummy private messages.  $\blacksquare$

{\bf Nested Optimality of Private Messages}: 
Nested optimality condition states that, in solving \ref{L1} to \ref{L13} for a given 
$ \mu$, the power of private messages will be at a saturation level which depends on $\mu$. 
This means, beyond the point of saturation, increasing the power of public message(s) will result in a higher return in terms of the weighted sum-rate. 

In other words, if power values allocated to private messages are equal to 
$\hat{P}_1(\mu)$ and $\hat{P}_2(\mu)$, solving \ref{L1} to \ref{L13} for 
$P_1=\hat{P}_1(\mu)$, $P_2=\hat{P}_2(\mu)$ 
results in zero power to be allocated to public messages. Likewise, solving \ref{L1} to \ref{L13} for 
$P_1=\hat{P}_1(\mu)$ allocates zero power to the public message of user 1, and solving \ref{L1} to \ref{L13} for 
$P_2=\hat{P}_2(\mu)$ allocates zero power to the public message of user 2.

As indicated, saturation levels for private messages depend on $\mu$.
All discussions presented hereafter in conjunction with saturation levels concern the same value of $\mu$ determining the relevant saturation levels. 
 One conclusion is that, in solving \ref{L1} to \ref{L13}, public messages are formed above private messages and should be decoded first while considering private messages as noise. Later parts show that the dummy private messages are placed first. Following discussions are based on such a configuration. 

Let us consider the multiple-access region formed over redundant message of user 2 
and non-redundant message of user 1 at $Y_1$, called maC1, and the multiple-access region formed over redundant message of user 1
and non-redundant message of user 2 at $Y_2$, called maC2. 
The optimum solution to \ref{L1} to \ref{L13} 
maximizes $R_{V_1}$ over maC1, and $R_{V_2}$ over maC2. Noting nested optimality, we can focus on the optimality of maC1 and maC2 assuming $P_1=\hat{P}_1(\mu)$, $P_2=\hat{P}_2(\mu)$, i.e., set aside the roles of public messages in optimizing maC1 and maC2. We have the following theorem.
\begin{theorem}
\label{Th2}
Consider the optimum solution to maC1 and maC2 in conjunction with an AWGN channel. 
The optimum solution maximizes $R_{V_1}$ over maC1 and $R_{V_2}$ over maC2. 
The dummy private messages should have the smallest rate in their corresponding multiple-access region, and should be ranked first (decoded last in successive decoding), with non-redundant private messages positioned next.   
 In addition, if, in the final solution, the rates of redundant and non-redundant private messages,  in maC1 and maC2 (which are functions of $a$, $b$, $\mu$ only) are equal to $\breve{r}_1$, $\breve{r}_2$,  $\hat{r}_1$ and $\hat{r}_2$, respectively, then solving \ref{L1} to \ref{L13} subject to some of the following constraints
$$
R_{V_1}\neq \hat{r}_1~~R_{V_2}\neq \hat{r}_2~~R^{(1)}_{V_2}\neq \breve{r}_1~~
R^{(2)}_{V_1}\neq \breve{r}_2
$$
 reduces the value of $L_0$ in the optimum solution. 
\end{theorem}

\begin{proof}
Noting that the weight of redundant (dummy) private messages is zero in $L_0$ defined in \ref{L1}, which is the smallest among different weights forming $L_0$, it follows that the rates of each redundant private message should be the minimum in its respective multiple-access region.  Under this condition, let us consider all possible decoding orders in maC1 and maC2. It follows that the only possible decoding order which can satisfy nested optimality condition, satisfy consistency condition in \ref{L4} to \ref{L13} while achieving a point on the sum-rate front of both multiple-access regions, is the one that places dummy private messages first (decoded last in successive decoding). Achieved corner point in maC1 is  based having the non-redundant message of user 1 above the noise plus interference (caused by the dummy private message of user 2). Likewise, achieved corner point in maC2 is  based having the non-redundant message of user 2 above the noise plus interference (caused by the dummy private message of user 1). Both these corner points correspond to a stationary solution in terms of power values in the sense that changing (increasing or decreasing) the power of one of the private messages results in a reduction in the rate of the other private message.  
\end{proof}

In summary, in the optimum solution, the rates of redundant and non-redundant private messages in maC1 and maC2 are fixed (function of $\mu$ and gain values $a$ and $b$). Optimality in this case translates to achieving the designated rates with minimum power, while maintaining consistency in the sense of \ref{L4} to \ref{L13}. 
$\blacksquare$

The following Theorem, although not directly related to 2-users GIC, sheds some light on the reason for relying on Gaussian random code-books. 
\begin{theorem}
\label{Gau}
Consider a two-level superposition code with AWGN. The rates of code-books should be equal to $\tilde{r}_1$ and $\tilde{r}_2$. The optimum solution achieving the minimum power vector relies on Gaussian code-books for both layers. 
\end{theorem}

\begin{proof}
 A two-layer Gaussian code-book result in minimizing the total power for a sum-rate of $\tilde{r}_1+\tilde{r}_2$. The power associated to the first code-book is minimum due to forming a Gaussian code-book subject to noise only. On the other hand, the sum of power values is determined by the sum-rate and is minimized using a Gaussian code-book of total power $\breve{P}$ satisfying $\tilde{r}_1+\tilde{r}_2=0.5\log_2(1+\breve{P})$. Noting the sum-power is minimized and the power of first layer is minimized as well, it follows that the power of second layer is minimized too. This problem has two solutions, depending on the order of code-books realizing rate $\tilde{r}_1$ and $\tilde{r}_2$. 
\end{proof}

\begin{theorem}
\label{Gau2}
Consider maC1 and maC2 with AWGN. The rates of code-books forming maC1 as well as maC2 are fixed (function of $\mu$ and channel gain values $a$ and $b$).  The optimum solution minimizing power for the designated rates rely on Gaussian code-books where the code-books corresponding to dummy private messages are ranked first.
\end{theorem}

\begin{proof}
According to Theorem~\ref{Th2}, dummy private messages are placed first at both $Y_1$ and $Y_2$. Regardless of the random coding density function, rate of these first layers (dummy private messages) will be an increasing function of their respective power levels. Saturation, which is the cause of nested optimality, occurs at a point that further increase in the power of one or both of private message(s) result(s) in a lower increase in the weighted sum-rate as compared to allocating any extra power, in one or both transmitter(s), to the respective public message(s). This means, saturation happens when the rate of one or both dummy private message(s) has/have reached to a respective threshold which is a function of  $\mu$ and channel gain values $a$ and $b$. Let us refer to these saturation rates as $\breve{r}_1$ and $\breve{r}_2$, respectively. Noting the total power at $Y_1$ and $Y_2$ are equal to $\hat{P}_1(\mu)+a\hat{P}_2(\mu)$ and
$\hat{P}_2(\mu)+b\hat{P}_1(\mu)$, respectively, the problem reduces to maximizing the rate of each non-redundant private message subject to the condition that the rate of its counterpart dummy private message is fixed. This situation corresponds to the conditions of Theorem \ref{NNN2} at $Y_1$ and $Y_2$, which states the optimum code-books will be independent Gaussian.  

The proof for Theorem \ref{Gau2} could also follow from Theorem~\ref{Gau}. Note that in maC1 and maC2, minimizing the power of each dummy private message is equivalent to minimizing the power of its associated non-redundant private message since the power values are linked to each other through \ref{L6} and \ref{L7}. 
\end{proof}

In Theorem~\ref{Gau2}, rate of dummy private messages are considered as bottlenecks governing the saturation conditions. A question may arise why the bottleneck is not directly related to the rates of non-redundant private messages. To explain the reason, let us consider adding an extra $\delta$-layer of power (see \cite{HK1} or Section~\ref{sec1.2} of the current report for definition) beyond the point of saturation. Without loss of generality, let us focus on user 1. As shown in Fig.~\ref{Fig-new}, the added $\delta$-layer becomes part of superposition code structure forming the non-redundant private message of user 1. The only distinction is that the rate of this added $\delta$-layer, unlike layers before it, is determined such that the layer can be decoded at both $Y_1$ and $Y_2$, i.e., it becomes part of the public message of user 1. However, in the formation of dummy private message of user 1 at $Y_2$, the added $\delta$-layer leaves the continuum of superimposed $\delta$-layers forming the dummy private message of user 1, and will be placed above the non-redundant private message of user 2. This entails the structure of dummy private message of user 1 at $Y_2$ is locked at the saturation level and its rate does not change beyond the point of saturation. According to Theorem 3 of \cite{HK1}, the rate of the added $\delta$-layer is governed by the structure of code-books at $Y_2$. As far as $Y_1$ is concerned, the rate of the added layer follows the smooth (infinitesimal) decrease in rates of subsequent $\delta$-layers stacked together to form the non-redundant private message of user 1 at $Y_1$, while the rate of added $\delta$-layer at $Y_2$ experiences a much larger decrease with respect to its previous layer placed on top of the redundant message of user 1 at $Y_2$.   This is in accordance with the rate of  non-redundant private message of user 1 to remain locked at $\breve{r}_1$ beyond the point of saturation.      

  \begin{figure}[h]
   \centering
   \includegraphics[width=0.7\textwidth]{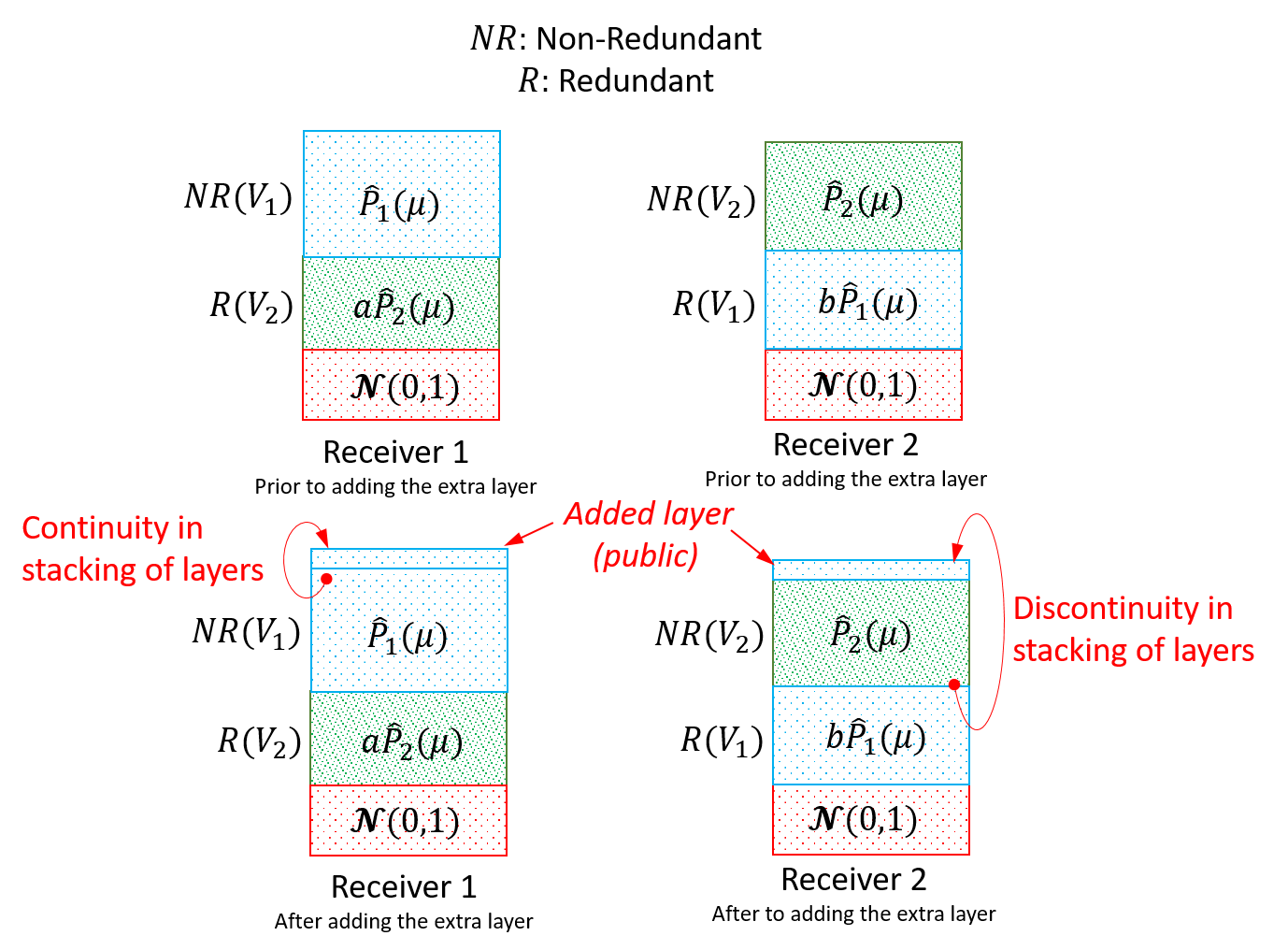}
   \caption{Conditions governing nested optimality of private messages, and addition of a $\delta$-layer beyond point of saturation. }
   \label{Fig-new}
 \end{figure}

In summary, nested optimality entails the optimum values for rates of redundant and non-redundant private messages, forming maC1 and maC2, should be realized with minimum power using Gaussian random code-books for private messages.

\begin{theorem}
\label{Pub}
Random coding density functions for public messages should be Gaussian. 
\end{theorem}

\begin{proof}
The proof follows noting public code-books see a Gaussian noise at both $Y_1$ and $Y_2$. Public messages form a multiple-access channel, called Mac1, at $Y_1$ and a multiple-access channel, called Mac2, at $Y_2$. Gaussian random code-books are optimum in both cases. Indeed, public code-book at $Y_1$ can be viewed as a Gaussian layer (on top of private code-book of user 1) in a superposition code encoding message of user 1. Likewise, public code-book at $Y_2$ can be viewed as a Gaussian layer (on top of private code-book of user 2) in a superposition code encoding message of user 2. However, rates of public layers are adjusted to be decodable at both receivers. This is needed to provide a space of shared rate values between $Y_1$ and $Y_2$ to be able to satisfy conditions over both $\overline{M\!AC_1}$ and $\overline{M\!AC_2}$ in \ref{L1} to \ref{L13}. 
\end{proof}

{\bf Example 1:} Consider the problem of designing a two-levels superposition code over an AWGN channel subject to an available power of $\breve{P}$. The objective is to maximize the rate of one of the layers, subject to having a rate of $r_1$ for the other layer. There are two solutions for this problem resulting in the same value of $S-r_1$ for the objective function: 
(i) layer of rate $r_1$ is first, and a layer of rate $S-r_1$ is second, where $S=0.5\log_2(1+P/\sigma^2)$ is the achievable sum-rate. 
(ii)  layer of rate $S-r_1$ is first, and layer of rate $r_1$ is second. 

In both solutions, the random coding density functions for the two layers are independent Gaussian. Either of these solutions may appear counter intuitive. In solution (i), a Gaussian code-book (that does not count in the objective function) is used in the first layer, creating the worst noise (Gaussian) for the second code-book above it, while this second code-book is the one  that determines the value of objective function. In solution (ii), second code-book is selected to be Gaussian. This maximizes the rate of the second code-book, while this rate does not count in the objective function.  
To explain these contradictions, in both cases,  Gaussian code-books play the important role of minimizing the power required for achieving their respective rates and thereby leave a higher (maximum) amount of power for the other code-book. It is a property of multiple access channel that saving in power due to using a Gaussian code-book is advantageous in comparison to  the harm caused by creating the worst (Gaussian) interference. 
More generally, in optimizing a weighted sum-rate over a multiple-access channel (or multi-level coding in a point-to-point setup), all layers, regardless of their weights (even if their respective weight is zero) should rely on  Gaussian random code-book. Of course, if the power constraints permit, zero power will be allocated to code-books that do not count (have zero weight) in the objective function. However, in some cases, including the problem arising from 2-users GIC, power values cannot be set to zero. In particular, the power of each dummy private messages is linked to the power of its relevant non-redundant private message through \ref{L6} and \ref{L7}. In this case, power values allocated to dummy private messages are minimized subject to problem constraints in \ref{L1} to \ref{L13}. As expressions in \ref{M1} are in terms of rates, the rates of dummy private messages is what contributes to the formation of the optimum solution. To further reduce the power of dummy private messages for a given rate, these are placed first which results in highest return in rate vs. allocated power. 

Returning to our example, consider solution (ii) and assume in the final solution, one assigns a rate $\breve{r}>r_1$ to the first layer without changing the power allocation. Obviously, the first layer cannot be decoded any longer. However, the second layer above it can still achieve its assigned rate of $S-r_1$ simply because it sees the first layer merely as interference (rate of the first layer does not come to the picture) and as long as the power does not change, the decoding can proceed. This is analogous what happens in 2-users GIC. We start by including dummy rate values which have a zero weight in the weighted sum-rate, and then conclude to decode the code-books which do count in the weighted sum-rate, decoding of these dummy messages is not needed.  More importantly, regardless of which code-books count in the weighted sum-rate and which code-books have zero weight, random coding density for all code-books should be independent Gaussian.$\blacksquare$.

\section{Optimum Solution to 2-users GIC}

Consider a 2-users GIC and rely on the concept of dummy private messages, which entails interfering terms can be independently encoded/decoded (subject to power constraints in \ref{L4} to \ref{L9}).  The weighted sum-rate achievable in such a relaxed 2-users GIC cannot be worse that the weighted sum-rate achievable in the original 2-users GIC. On the other hand, the rate-tuple achieved at each receiver in the relaxed 2-users GIC cannot be better than what can be achieved in a multiple-access channel which is optimally enlarged (using single letter, independent Gaussian code-books forming a superposition code at each transmitter). This means solution to \ref{L1} to \ref{L13} provides an upper bound on the solution to 2-users GIC. To explain this point differently, let us consider the optimum solution to 2-users GIC. Interference terms at each receiver can be assigned a rate value that renders them decodbale when messages above interference terms are decoded and removed. Including this dummy rates for interference terms result in a 4-tuple rate vector at $Y_1$ and another 4-tuple rate vector at $Y_2$ that cannot fall outside  $\overline{M\!AC}_1$ and $\overline{M\!AC}_2$, respectively.

Once dummy private messages are ranked first, if they are decodable or not, does not affect the rate-tuple achievable by code-books above them. The impacts of dummy private messages on code-books above them merely depend on their respective power. This property permits the rate-tuple achieved in 2-users GIC (over non-redundant messages, without decoding dummy private messages) to be the same as the rate-tuple achieved in the lower bound corresponding to HK rate region, and the upper bound based on the intersection of $\overline{M\!AC}_1$ and $\overline{M\!AC}_2$ in  \ref{L1} to \ref{L13}. Noting properties of multiple-access channel, any weighted sum-rate is optimized using single letter, independent Gaussian code-books forming a superposition code at each transmitter. 
 
Overall orderings of messages are as follows: 

\vspace{0.2cm}
\noindent
\begin{equation}
\label{T1}
\begin{small}
\overline{M\!AC}_1: \mbox{(1) Dummy Private Message of User 2, (2) Private Message of User 1, (3,4) Public Messages} 
\end{small}
\end{equation}
\begin{equation}
\label{T2}
\begin{small}
\overline{M\!AC}_2: \mbox{(1) Dummy Private Message of User 1, (2) Private Message of User 2, (3,4) Public Messages} 
\end{small}
\end{equation}

\noindent
Using this  ordering, the rate-tuple of non-redundant messages in the optimum solution to \ref{L1} to \ref{L13} (with or without decoding of dummy private messages) is  the same as the rate-tuple achieved in HK optimization problem defined in \ref{HK0p} to \ref{HK16p} (without decoding the dummy private messages), as well as the rate tuple achieved in 2-users GIC (without decoding the dummy private messages).

\section{A Closer Look at the Optimum Solution} \label{sec3}

This section helps in understating arguments discussed in earlier parts, and paves the way to understanding the structure of the final solution.  As will be discussed later, the formation of optimum solution to 2-users GIC involves three corner points in $\overline{M\!AC_1}$ and $\overline{M\!AC_2}$, two on the dominant MAC involved in time-sharing and one on non-dominant MAC. Theorem \ref{LL1} below clarifies that in a nested formation of Multiple-Access Polymatroids (MAPs), a corner point on a parent MAP corresponds to a corner point in each of its children. 

\begin{theorem}
\label{LL1}
Consider a Multiple-Access Polymatroid (MAP) composed of $M$ messages ordered from bottom to top indexed by $1, 2,\ldots, M$. Consider integer 
$1\leq \ell \leq M$. Two MAPs are identified, 
one including messages $\ell+1,\ldots,M$ by considering interference from $1,\ldots,\ell$ as noise, and the other one composed of messages $1,\ldots,\ell$ when messages $\ell+1,\ldots,M$ are decoded and removed. The former is denoted as Map and the latter as maP. We have: (i) In the sense of iff,  a corner point (on the sum-rate front) of MAP corresponds to a corner point in Map and a corner point in maP. (ii) Map is equal to the projection of MAP on the orthogonal span of rate tuple  corresponding to maP, i.e., Map is lifted to MAP by including the missing rate-tuples from maP. 
\end{theorem}
\begin{proof}
Proof follows noting that expressions for Map are a subset of expressions defining MAP. Direct replacement shows that region Map is mapped to a corner region in MAP by including missing rate components by rate-tuples from maP. 
 \end{proof}

Theorem \ref{Th1} below is a different way of expressing the fact that public messages are formed above private messages. In the language of nested optimality, public messages are formed when private messages have reached their saturation level (which depends on the value of $\mu$, and channel gains, i.e., $a$ and $b$).
  
\begin{theorem} \label{Th1}
Consider a 2-users GIC (see Fig.~\ref{fig1}) realizing a rate-tuple $(R_{U_1},  R_{U_2}, R_{V_1}, R_{V_2})$. In realizing $(R_{U_1},  R_{U_2}, R_{V_1})$ at $Y_1$, $V_2$ is considered as noise, and likewise, in realizing $(R_{U_1},  R_{U_2}, R_{V_2})$ at $Y_2$, $V_1$ is considered as noise. 
\end{theorem}

\begin{proof}
Without loss of generality, let us focus on $Y_1$, and let us assume part of message $V_2$, denoted as $\hat{V}_2$, is decoded at $Y_1$. Noting that $V_2$ is entirely decoded at $Y_2$, it turns out that $\hat{V}_2$ is decoded at both $Y_1$ and $Y_2$, and consequently, it will be a public message. By moving $\hat{V}_2$ from private to public (and accordingly readjusting power allocation), the private message of user 2 reduces until it entirely acts as noise in $Y_1$.   
\end{proof}

Another question concerns the reason for having a private message and a public message for each user in 2-users GIC, which has motivated the formulation in \ref{L1} to \ref{L13}. Having superposition coding at $X_1$ and $X_2$, equipped with power allocation between layers,  
does not contradict optimality. Rates of public and private messages are adjusted to enable simultaneously satisfying conditions associated with multiple-access reception at $Y_1$ and $Y_2$
(also see~\cite{HK1}). 

\subsection{Position of Optimum Solution on the Sum-rate Front of Dominant MAC}

It is also known that formulation in \ref{HK0p} to \ref{HK16p}, corresponding to HK rate region, results in an achievable solution (lower bound) to 2-users GIC. Upon projecting  $\overline{M\!AC_1}$ and  $\overline{M\!AC_2}$ along the orthogonal span of their respective dummy private message, one obtains MAC1 (defined in HK constraints \ref{HK1p} to \ref{HK7p})  and MAC2 (defined in HK constraints \ref{HK8p} to \ref{HK14p}). The optimum solution to  \ref{HK0p} to \ref{HK16p} falls on the sum-rate of MAC1 and/or on the sum-rate front of MAC2. Otherwise, the solution would be within both regions, MAC1 and MAC2, and could be improved by increasing some of its (non-redundant) rate values. Let us assume the optimum solution falls on the sum-rate front of MAC1, called dominant. It is easy to see that the solution corresponds to a corner point in MAC2 (called non-dominant) once one of its corresponding public rates is reduced to bring the solution to a point on the sum-rate front of MAC1. Any point on the sum-rate front of MAC1 can be realized by time-sharing between two corner points, called time-sharing corner points hereafter. The optimum solution, falling on the sum-rate front of MAC1 (dominant) by interpolation between time-sharing  corner points, is lifted to a point on the sum-rate front of $\overline{M\!AC_1}$ by including the rate of the dummy private message of user 2.  Likewise, the solution on the boundary  of MAC2  (non-dominant) is lifted to a point  on the boundary of $\overline{M\!AC_2}$. This boundary point is  obtained by reducing the rate of a public message with respect to a corner point of $\overline{M\!AC_2}$. 

Rates of public messages are the only rate values that are shared between $\overline{M\!AC_1}$/MAC1/Mac1 and  $\overline{M\!AC_2}$/MAC2/Mac2. As a result, satisfying both constraints (in \ref{L2} and \ref{L3}) requires adjusting the rates of public messages, including: (i) time-sharing between two corner points in the dominant multiple-access region defined over public messages, and reducing the rate of one of public messages in non-dominant multiple-access region defined over public messages  such that the final solution falls on the sum-rate front of the dominant multiple-access region.  Noting Theorem~\ref{LL1}, a corner point in $\overline{M\!AC_1}$ corresponds to a corner point in MAC1 (region of non-redundant messages at $Y_1$), 
a corner point in
Mac1 (region of public messages at $Y_1$), and  a corner point in
maC1 (region of private messages at $Y_1$). Likewise, a corner point in $\overline{M\!AC_2}$ corresponds to a corner point in MAC2 (region of non-redundant messages at $Y_2$), 
a corner point in Mac2 (region of public messages at $Y_2$), and  a corner point in
maC2 (region of private messages at $Y_2$). As we will be discussed in Section \ref{sec3}, at most three corner points will contribute to the final solution, two on the dominant MAC which will be involved in time-sharing and one on the non-dominant MAC. The optimum solution always falls on corner points of maC1 and maC2, but, assuming $\overline{M\!AC_1}$ is dominant,  it falls on the sum-rate front of $\overline{M\!AC_1}$/MAC1/Mac1. Note that region of public messages is the only region shared between $\overline{M\!AC_1}$  and $\overline{M\!AC_2}$, and time-sharing over it enables realizing a point that satisfies both constraints  (related to $\overline{M\!AC_1}$  and $\overline{M\!AC_2}$), while being optimized through power allocation.
 
\subsection{Lack of Ability to Decoded  Interference Terms - Impact on Optimizing Non-redundant Rate Values}

In the case of 2-users GIC,  the code-books forming interference terms are the same as the code-books forming non-redundant private messages. These changes in code-books/rates with respect to dummy messages in $\mathcal{P}_0$ do not affect the solution in maximizing the weighted sum-rate in 2-users GIC. Note that in successive decoding, the code-books corresponding to non-redundant messages are above the interference terms, and consequently, see the interference as additional noise. The important point is that the power of interference terms are minimized, providing the means to maximize the weighted sum-rate over non-redundant messages. This point establishes the link between the optimum solution to 2-users GIC and the optimum solution to problem $\mathcal{P}_0$ in \ref{L1} to \ref{L13}. 

The final conclusion is that, in the optimum solutions explained above, the lower bound (HK rate region in  \ref{HK0p} to \ref{HK16p}) and the upper bound (problem $\mathcal{P}_0$ defined in \ref{L1} to \ref{L13}) coincide in the space spanned by non-redundant rate values, concluding that both result in the optimum solution to 2-users GIC. 

\section{Relationship to HK Rate Region}
To derive the HK constraints using notations widely used in the literature,  one should refer to expressions 3.2 to 3.15 on page 51 of~\cite{HK5}, and apply the following changes,
($\mbox{current~article} \leftrightarrow$~\cite{HK5}):
$U_1 \leftrightarrow W_1$, 
$U_2 \leftrightarrow W_2$,
$V_1 \leftrightarrow U_1$, 
$V_2 \leftrightarrow U_2$,
$R_{U_1}\leftrightarrow T_1$, $R_{U_2}\leftrightarrow T_2$, 
$R_{V_1}\leftrightarrow S_1$, $R_{V_2}\leftrightarrow S_2$. Applying these changes, the expanded Han-Kobayashi constraints  and the associated optimization problem  are expressed as follows

\begin{align} 
\label{HK0p}
\mbox{Maximize:}~~~~~ & R_{ws}= R_1+\mu R_2=R_{U_1}+ R_{V_1} + \mu(R_{U_2}+ R_{V_2})  \\
\mbox{Subject to:}~~~ &   \nonumber \\ \label{HK1p}
 R_{U_1}    & ~~{\le}~~   I(U_1;Y_1|U_2,V_1)     \\ \label{HK2p}
R_{U_2}   & ~~{\le}~~    I(U_2;Y_1|U_1,V_1)    \\ \label{HK3p}
R_{V_1}  & ~~{\le}~~  I(V_1;Y_1|U_1,U_2)   \\ \label{HK4p}
R_{U_1}+R_{U_2}   & ~~{\le}~~ I(U_1,U_2;Y_1|V_1)    \\  \label{HK5p}
R_{U_1}+R_{V_1}  & ~~{\le}~~  I(U_1,V_1;Y_1|U_2)    \\ \label{HK6p}
R_{U_2}+R_{V_1}   & ~~{\le}~~  I(U_2,V_1;Y_1|U_1)   \\ \label{HK7p} 
R_{U_1}+R_{U_2}+ R_{V_1}  & ~~{\le}~~   I(U_1,U_2,V_1;Y_1)    \\ \nonumber 
------- & ---------  \\  \label{HK8p}
 R_{U_1}  & ~~{\le}~~   I(U_1;Y_2|U_2,V_2)    \\ \label{HK9p}
 R_{U_2}   & ~~{\le}~~   I(U_2;Y_2|U_1,V_2)    \\  \label{HK10p}
 R_{V_2}  & ~~{\le}~~   I(V_2;Y_2|U_1,U_2)    \\ \label{HK11p}  
 R_{U_1}+R_{U_2}   & ~~{\le}~~   I(U_1,U_2;Y_2|V_2)  \\  \label{HK12p}
 R_{U_2}+R_{V_2}   & ~~{\le}~~   I(U_2,V_2;Y_2|U_1)   \\ \label{HK13p}
R_{U_1}+R_{V_2}    & ~~{\le}~~   I(U_1,V_2;Y_2|U_2)   \\  \label{HK14p}
 R_{U_1}+R_{U_2} + R_{V_2}    & ~~{\le}~~   I(U_1,U_2,V_2;Y_2) \\ \label{HK15p}
E(U_1^2)+E(V_1^2)& ~~=~~   P_1  \\  \label{HK16p}
E(U_2^2)+E(V_2^2) & ~~=~~   P_2.
\end{align}

To establish the connection between solution to problem $\mathcal{P}_0$ in \ref{L1} to \ref{L13} to HK optimization problem in \ref{HK0p} to \ref{HK16p}, without loss of generality, let us focus on constraints for $\overline{M\!AC_1}$. 
It includes fifteen constraints in terms of four rate values 
$(R_{U_1}, R_{V_1}, R_{U_2}, R^{(1)}_{V_2})$. A subset of constraints are as follows
 \begin{align} 
\label{HK1z}
 R_{U_1}    & ~~{\le}~~   I(U_1;Y_1|U_2,V_1,V_2)     \\ \label{HK2z}
R_{U_2}   & ~~{\le}~~    I(U_2;Y_1|U_1,V_1,V_2)    \\ \label{HK3z}
R_{V_1}  & ~~{\le}~~  I(V_1;Y_1|U_1,U_2,V_2)   \\ \label{HK4z}
R_{U_1}+R_{U_2}   & ~~{\le}~~ I(U_1,U_2;Y_1|V_1,V_2)    \\  \label{HK5z}
R_{U_1}+R_{V_1}  & ~~{\le}~~  I(U_1,V_1;Y_1|U_2,V_2)    \\ \label{HK6z}
R_{U_2}+R_{V_1}   & ~~{\le}~~  I(U_2,V_1;Y_1|U_1,V_2)   \\ \label{HK7z} 
R_{U_1}+R_{U_2}+ R_{V_1}  & ~~{\le}~~   I(U_1,U_2,V_1;Y_1|V_2). 
\end{align}
Constraints \ref{HK1z} to  \ref{HK7z} above, corresponding to constraints \ref{HK1p} to  \ref{HK7p} in HK constraints, define the multiple access region corresponding to treating the dummy private message of user 1 as well as the non-redundant private messages of user 2 as noise at $Y_1$. A similar set of equations can be obtained in conjunction with $\overline{M\!AC_2}$ coinciding with constraints \ref{HK8p} to  \ref{HK14p} in HK constraints, forming the multiple access region corresponding to treating the dummy private message of user 2 as well as the non-redundant private messages of user 1 as noise at $Y_2$

  \begin{figure}[h]
   \centering
   \includegraphics[width=0.7\textwidth]{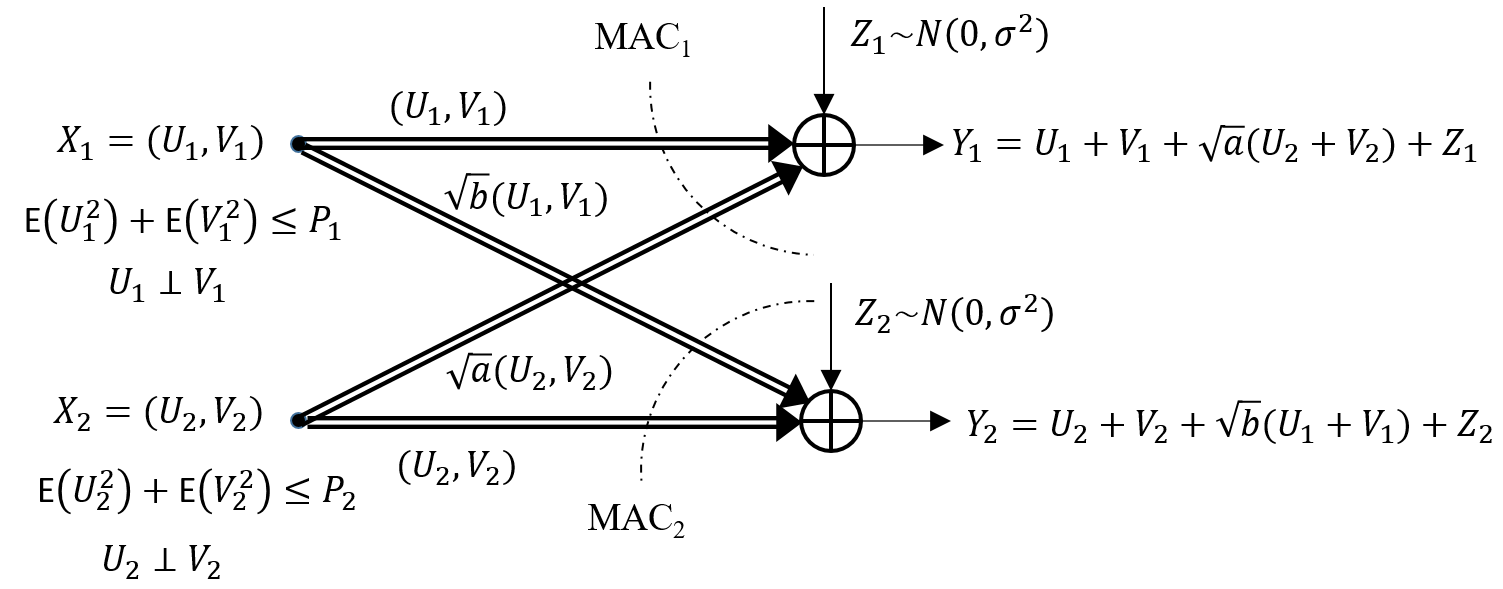}
   \caption{Model for 2-users Gaussian Interference Channel (GIC). }
   \label{fig1}
 \end{figure}

\subsection{Projection of a MAC Region on its Sub-spaces}

Consider a memory-less multiple-access channel where the mutual information terms on the right hand sides of corresponding constraints are some function of a power vector $\mathbf{p}$. In the following, the ploymatroid corresponding to capacity region of such a multiple access channel is simply called MAP, sanding for Multiple-Access Polymatroid. 
 
Let us consider a MAP with $M$ inputs, and focus on one of its rate values, say $r_1\in[\min r_1~\mbox{to}~\max r_1]$. For a given 
power vector, $\mathbf{p}$, value of $r_1$ depends on its order in successive decoding. $r_1$ is maximized when its corresponding code-book is placed first (decoded last) in successive decoding, and it is minimized when its corresponding code-book is placed last (decoded first) in successive decoding. Setting $r_1=\max r_1$ results in a MAP, $\mathsf{P}_{\max}$, defined over the space orthogonal to $r_1$, shifted up to position of $\max r_1$ along the $r_1$ axis. The constraints  for $\mathsf{P}_{\max}$ are a subset of constraints defining the original MAP. These constraints in the original MAP correspond to placing the $r_1$ code-book first, decoding all code-books above the $r_1$ code-book, while treating the interference from lower code-books (including $r_1$ code-book) as noise. $r_1$ plays the role of dummy private message in $\overline{M\!AC_1}$ and $\overline{M\!AC_2}$.

Likewise,  setting $r_1=\min r_1$ results in a MAP, $\mathsf{P}_{\min}$, 
defined over the space orthogonal to $r_1$, shifted up to position of $\min r_1$ along the $r_1$ axis.
 
Depending on the order of code-book corresponding to $r_1$ in successive decoding, one obtains $M$ cuts (cuts are orthogonal to $r_1$-axis) in the original MAP region.  
The projections of these cuts on the sub-space orthogonal to $r_1$ results in a nested set of regions bounded within an interior MAP and an exterior MAP. Interior MAP corresponds to $\max r_1$ and exterior MAP corresponds to $\min r_1$.  
Achieving capacity in a 2-users interference channel, formed by a MAP at $Y_1$ and a MAP at $Y_2$, narrows down to optimally enlarging the interior core formed when 
$r_1$ maximizes the minimum rate in the original MAP.    
Figure~\ref{fig2} provides an example. 
  \begin{figure}[htbp]
   \centering
   \includegraphics[width=0.5\textwidth]{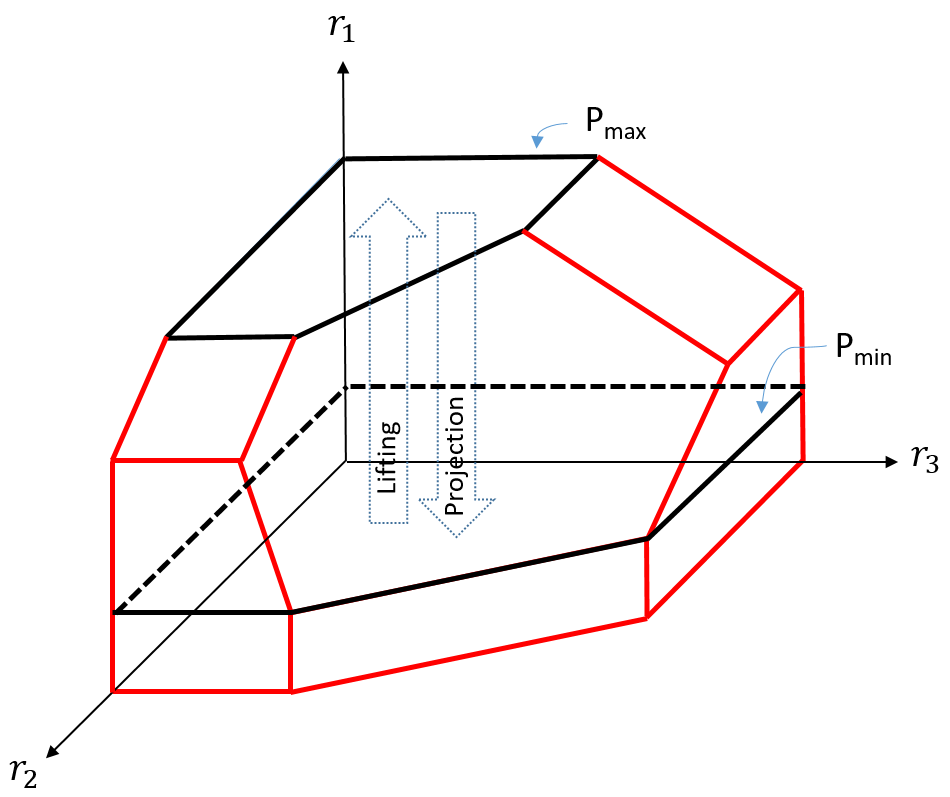}
   \caption{Example of nested projections of a three dimensional multiple access region with rate-tuple $(r_1,r_2,r_3)$ on the sub-space $(r_2,r_3)$.}
   \label{fig2}
 \end{figure}

\end{document}